\documentclass[a4paper,10pt]{amsart}
\usepackage{graphicx}
\usepackage{amssymb,amsmath,amstext,amscd}

\theoremstyle{plain}
\newtheorem{thm}{Theorem}[section]

\newtheorem*{main theorem}{Theorem}

\theoremstyle{definition}

\begin{document}

\title{The group of causal automorphisms }
\author{Do-Hyung Kim}
\address{Department of Applied Mathematics, College of Applied Science, Kyung Hee University,
Seocheon-dong, Giheung-gu,Yongin-si, Gyeonggi-do 446-701, Republic
of Korea} \email{mathph@khu.ac.kr}


\begin{abstract}
The group of causal automorphisms on Minkowski space-time is given
and its structure is analyzed.
\end{abstract}

\maketitle

\section{Introduction} \label{section: 1}

In 1964, Zeeman has shown that any causal automorphism on
Minkowski space-time $\mathbb{R}^n_1$ can be represented by a
composite of orthochronous transformation, traslation and
homothety when $n \geq 3$.(See Ref. \cite{Zeeman}). Therefore, if
we let $G$ be the group of all causal automorphisms on
$\mathbb{R}^n_1$ with $n \geq 3$, then $G$ is isomorphic to the
semi-direct product of $\mathbb{R}^+ \times \mathcal{O}$ and
$\mathbb{R}^n_1$ where $\mathbb{R}^+$ is the group of positive
real numbers and $\mathcal{O}$ is the group of orthochronous
transformations on $\mathbb{R}^n_1$. However, as he remarked,
Zeeman's theorem does not hold when $n=2$.

In connection with this, recently, the standard form of causal
automorphism on $\mathbb{R}^2_1$ is given in Ref. \cite{CQG4}. In
this paper, we improve the result in Ref. \cite{CQG4} and by use
of this, the structure of the causal automorphism group is
analyzed when $n=2$.

\section{Causal automorphisms on $\mathbb{R}^2_1$} \label{section: 2}

The following is a known result.

\begin{thm} \label{standard}

Let $F : \mathbb{R}^2_1 \rightarrow \mathbb{R}^2_1$ be a causal
automorphism. Then, there exist a continuous functions $g :
\mathbb{R} \rightarrow \mathbb{R}$ and a homeomorphism $f :
\mathbb{R} \rightarrow \mathbb{R}$ which satisfy $\sup(g \pm
f)=\infty$, $\inf (g \pm f) = -\infty$ and $| \frac{g(t+\delta
t)-g(t)}{f(t+\delta t)-f(t)}| < 1$ for all $t$ and $\delta t$,
such that if $f$ is increasing,
then $F$ is given by \\
 $F(x,y) = ( \frac{f(x-y)+f(x+y)}{2} +
\frac{g(x+y)-g(x-y)}{2} , \frac{f(x+y)-f(x-y)}{2} +
\frac{g(x-y)+g(x+y)}{2} )$\\
and if $f$ is decreasing, then we have \\
  $F(x,y) =
(\frac{f(x+y)+f(x-y)}{2}+\frac{g(x-y)-g(x+y)}{2},\frac{f(x-y)-f(x+y)}{2}+\frac{g(x+y)+g(x-y)}{2})$.\\
Conversely, for any functions $f$ and $g$ which satisfy the above
conditions, the function $F : \mathbb{R}^2_1 \rightarrow
\mathbb{R}^2_1$ defined as above, is a causal automorphism.

\end{thm}

\begin{proof}
See Theorem 4.4 in Ref. \cite{CQG4}.
\end{proof}

We remark that in the above theorem, for given causal automorphism
$F(x,y)=(F_1(x,y),F_2(x,y))$, the homeomorphism $f$ and the
continuous function $g$ are uniquely determined by $f(t)=F_1(t,0)$
and $g(t)=F_2(t,0)$. By use of this, we improve the above theorem.

\begin{thm} \label{main1}
Let $F : \mathbb{R}^2_1 \rightarrow \mathbb{R}^2_1$ be a causal
automorphism on $\mathbb{R}^2_1$. Then, there exist unique
homeomorphisms $\varphi$ and $\psi$ of $\mathbb{R}$, which are
either both increasing or both decreasing, such that if $\varphi$
and $\psi$ are increasing, then we have $F(x,y) =
\frac{1}{2}(\varphi(x+y)+\psi(x-y), \varphi(x+y)-\psi(x-y))$, or
if $\varphi$ and $\psi$ are decreasing, then we have $F(x,y) =
\frac{1}{2}(\varphi(x-y)+\psi(x+y), \varphi(x-y)-\psi(x+y))$.

Conversely, for any given homeomorphisms $\varphi$ and $\psi$ of
$\mathbb{R}$, which are either both increasing or both decreasing,
the function $F$ defined as above is a causal automorphism of
$\mathbb{R}^2_1$.
\end{thm}

\begin{proof}

For any given causal automorphism $F : \mathbb{R}^2_1 \rightarrow
\mathbb{R}^2_1$, by the previous theorem, we can get unique
homeomorphism $f$ and unique continuous function $g$. Let
$\varphi=f+g$ and $\psi=f-g$. Then, clearly, $\varphi$ and $\psi$
are continuous.

The conditions $\sup(g \pm f)=\infty$, $\inf (g \pm f) = -\infty$
imply that $\varphi$ and $\psi$ are surjective. If
$\varphi(t)=\varphi(t^\prime)$ with $t \neq t^\prime$, then this
implies that $\frac{g(t^\prime)-g(t)}{f(t^\prime)-f(t)}=-1$, which
contradicts to the condition $| \frac{g(t+\delta
t)-g(t)}{f(t+\delta t)-f(t)}| < 1$ for all $t$ and $\delta t$.
Therefore, $\varphi$ is injective and likewise, we can show that
$\psi$ is injective. Since $\varphi$ and $\psi$ are continuous
bijections from $\mathbb{R}$ to $\mathbb{R}$, the topological
domain of invariance implies that $\varphi$ and $\psi$ are
homeomorphisms.

The equation $\varphi(t_2)-\varphi(t_1) =
(f(t_2)-f(t_1))[1+\frac{g(t_2)-g(t_1)}{f(t_2)-f(t_1)}]$ tells us
that, by the condition $| \frac{g(t+\delta t)-g(t)}{f(t+\delta
t)-f(t)}| < 1$, $\varphi$ is increasing if and only if $f$ is
increasing. Likewise, we can show that $\psi$ is increasing if and
only if $f$ is increasing.

By simple calculation, we can show that $F$ has the desired form
when expressed in terms of $\varphi$ and $\psi$. This completes
the proof of the first part.

To prove the converse, let $\varphi$ and $\psi$ be increasing
homeomorphisms on $\mathbb{R}$. Let $f = \frac{1}{2}(\varphi +
\psi)$ and $g = \frac{1}{2}(\varphi - \psi)$. Then $g$ is
continuous and $f$ is an increasing homeomorphism. Since $\varphi$
and $\psi$ are homeomorphisms, we have $\sup(f+g) = \sup \varphi =
\infty$, $\sup(f-g) = \sup \psi = \infty$, $\inf(f+g) = \inf
\varphi = -\infty$ and $\inf(f-g) = \inf \psi = -\infty$.

We now show that $f$ and $g$ satisfy the inequality $|
\frac{g(t+\delta t)-g(t)}{f(t+\delta t)-f(t)}| < 1$. Let $\Delta =
\frac{g(t)-g(t_0)}{f(t)-f(t_0)} =
\frac{\varphi(t)-\varphi(t_0)-(\psi(t)-\psi(t_0))}{\varphi(t)-\varphi(t_0)+\psi(t)-\psi(t_0)}$.
Without loss of generality, we can assume $t > t_0$ and it is easy
to see that $-1 < \Delta <1$, since $\varphi$ and $\psi$ are
increasing. Therefore, by the theorem ~\ref{standard}, the
function $F$ defined as in the statement, is a causal
automorphism. By the exactly same argument, we can show that the
assertion also holds when $\varphi$ and $\psi$ are both decreasing
homeomorphisms.
\end{proof}

\section{The group of causal automorphisms on $\mathbb{R}^2_1$ } \label{section: 3}

In this section, we denote the group of all causal automorphisms
on $\mathbb{R}^2_1$ by $G$, and we analyze its group structure.
For this we let $H(\mathbb{R})$ be the group of all homeomorphisms
on $\mathbb{R}$ and let $H=H^+ \cup H^-$ where
$H^+=\{(\varphi,\psi) \in H(\mathbb{R}) \times H(\mathbb{R}) \,\,
| \,\, \varphi, \psi \,\, \mbox{are increasing}\, \}$ and
$H^-=\{(\varphi,\psi) \in H(\mathbb{R}) \times H(\mathbb{R}) \,\,
| \,\, \varphi, \psi \,\, \mbox{are decreasing}\, \}$. Then $H$ is
a subgroup of $H(\mathbb{R}) \times H(\mathbb{R})$ under the
operation induced from $H(\mathbb{R}) \times H(\mathbb{R})$.

From the Theorem ~\ref{main1}, we can see that any causal
automorphism $F$ on $\mathbb{R}^2_1$ corresponds to a unique
element in $H$ and, conversely, each elements in $H$ uniquely
determines a causal automorphism on $\mathbb{R}^2_1$. Thus, there
exists a one-to-one correspondence between $G$ and $H$ as a set.
It might seem that $G$ is isomorphic to $H$. However, we cannot
obtain an isomorphism in this way and we define a new operation on $H$ as follows.

We define a $\mathbb{Z}_2$-action on $H$ by $a \cdot (\varphi,
\psi) = (\varphi, \psi)$ if $a=0$ and $a \cdot (\varphi, \psi) =
(\psi, \varphi)$ if $a=1$. If we define a map $\pi : H \rightarrow
\mathbb{Z}_2$ by $\pi(x) = 0$ if $x \in H^+$ and $\pi(x)=1$ if $x
\in H^-$, then $\pi$ is a group homomorphism with $H$ equipped
with the operation induced from $H(\mathbb{R}) \times
H(\mathbb{R})$. Note that $\pi ( \varphi, \psi) =
\pi(\varphi^{-1}, \psi^{-1}) = \pi(\psi, \varphi)$ and $\pi (
\pi(a,b) \cdot (\varphi, \psi)) = \pi(\varphi, \psi)$ for any
$(a,b) \in H$.

To get an isomorphism from $G$ to $H$, we define a new operation $*$
on $H$ by $(\alpha, \beta) * (\varphi, \psi) = (\alpha, \beta)
\circ \pi(\alpha, \beta) \cdot (\varphi, \psi)$ where $\cdot$ is
the $\mathbb{Z}_2$-action defined above and $\circ$ is the
operation induced from $H(\mathbb{R}) \times H(\mathbb{R})$.

\begin{thm}
The set $H$ under $*$ is a group and is isomorphic to $G$.
\end{thm}
\begin{proof}

To show associativity, by calculation, we have, $(a,b) * \{(c,d) *
(e,f)\} = (a,b) * \{ (c,d) \circ \pi(c,d) \cdot (e,f) \} = (a,b)
\circ \pi(a,b) \cdot \{(c,d) \circ \pi(c,d) \cdot (e,f)\}$. Thus,
we
have the following four cases.\\

\begin{tabular}{|c|c|c|} \hline
$\pi(a,b)$ & $\pi(c,d)$ & $(a,b)*\{(c,d)*(e,f)\}$ \\ \hline 0 & 0
& $(a,b) \circ \{(c,d)\circ(e,f)\} = (ace,bdf)$ \\ \hline 0 & 1 &
$(a,b) \circ \{(c,d)\circ(f,e)\} = (acf, bde)$ \\ \hline 1 & 0 &
$(a,b) \circ (df, ce) = (adf, bce)$ \\ \hline 1 & 1 & $(a,b) \circ
(de, cf) = (ade, bcf)$ \\ \hline
\end{tabular}
\\
 On the other hand, we have
\begin{align*}
\{(a,b)*(c,d)\} *(e,f) &= \{(a,b) \circ \pi(a,b) \cdot (c,d)\} *
(e,f)  \\
&= \{(a,b) \circ \pi(a,b) \cdot (c,d)\} \circ \pi((a,b) \circ
\pi(a,b) \cdot (c,d)) \cdot (e,f)  \\
&= \{(a,b) \circ \pi(a,b) \cdot (c,d) \} \circ \{\pi(a,b)
\pi(\pi(a,b) \cdot (c,d))\} \cdot (e,f)
 \intertext{($\because$ $\pi$ is a group homomorphism.)}\\
 &= \{(a,b)
\circ \pi(a,b) \cdot (c,d)\} \circ \{
\pi(a,b)\pi(c,d)\} \cdot (e,f) \\
& \intertext{($\because \pi(\pi(a,b) \cdot (c,d)) = \pi(c,d)$)}
\end{align*}

Thus, we have the following four cases. \\

\begin{tabular}{|c|c|c|} \hline
$\pi(a,b)$ & $\pi(c,d)$ & $\{(a,b)*(c,d)\}*(e,f)$ \\ \hline
 0 & 0& $\{(a,b) \circ (c,d)\} \circ (e,f) = (ace,bdf)$ \\ \hline
  0 & 1 & $\{(a,b) \circ (c,d)\} \circ (f,e)=(acf,bde)$ \\ \hline
   1 & 0 & $\{(a,b) \circ (d,c)\} \circ (f,e) = (adf, bce)$ \\ \hline
   1 & 1 & $\{(a,b) \circ (d,c)\} \circ (e,f) = (ade, bcf)$ \\ \hline
\end{tabular}

The above two tables tell us that the operation $*$ satisfies the
associative law.

The following two formulae show that $(id, id)$ is the identity
element in $H$ under $*$.

\begin{align*}
(id, id) * (\varphi, \psi) &= (id, id) \circ \pi(id,id) \cdot
(\varphi, \psi) \\
&= (id, id) \circ (\varphi, \psi) = (\varphi, \psi),\\
(\varphi, \psi)*(id, id) &= (\varphi, \psi) \circ \pi(\varphi,
\psi) \cdot (id,id) \\
&= (\varphi, \psi) \circ (id,id) = (\varphi, \psi).
\end{align*}

The following two formulae show that, for given $(\varphi, \psi)
\in H$, $\pi(\varphi, \psi) \cdot (\varphi^{-1}, \psi^{-1})$ is
the inverse element of $(\varphi, \psi)$.

\begin{align*}
(\varphi, \psi)*\{\pi(\varphi, \psi) \cdot (\varphi^{-1},
\psi^{-1})\} &= (\varphi, \psi) \circ \pi(\varphi, \psi) \cdot \{
\pi(\varphi, \psi) \cdot (\varphi^{-1}, \psi^{-1}) \} \\
&= (\varphi, \psi) \circ \{ \pi(\varphi, \psi) \pi(\varphi, \psi)
\} \cdot (\varphi^{-1}, \psi^{-1})\\
 &= (\varphi, \psi) \circ (\varphi^{-1}, \psi^{-1}) = (id, id),\\
\{\pi(\varphi, \psi) \cdot (\varphi^{-1}, \psi^{-1}) \} *
(\varphi, \psi) &= \{ \pi(\varphi, \psi) \cdot (\varphi^{-1},
\psi^{-1})\} \circ \{\pi(\pi(\varphi, \psi) \cdot (\varphi^{-1},
\psi^{-1}))\} \cdot (\varphi, \psi) \\
&= \{\pi(\varphi, \psi) \cdot (\varphi^{-1}, \psi^{-1}) \} \circ
\{\pi(\varphi, \psi) \cdot (\varphi, \psi) \} = (id,id).
\end{align*}

Therefore, $H$ forms a group under the operation $*$.

We now show that the groups $G$ and $H$ under $*$ are isomorphic. For
given causal automorphism $F(x,y) = (F_1(x,y), F_2(x,y))$, Theorem
~\ref{main1} tells us that $F_1+F_2$ and $F_1-F_2$ determine
unique element in $H$. Therefore, we can define $\Pi : G
\rightarrow H$ by
$$ \Pi(F) = (\mbox{homeomorphism determined by} F_1+F_2, \mbox{homeomorphism determined
by}F_1-F_2) $$ By Theorem ~\ref{main1}, $\Pi$ is a bijection and
we only need to show that $\Pi$ is a homomorphism. We have the
following four cases.\\

i) Let
$G=\frac{1}{2}(\alpha(x+y)+\beta(x-y),\alpha(x+y)-\beta(x-y))$ and
$F=\frac{1}{2}(\varphi(x+y)+\psi(x-y), \varphi(x+y)-\psi(x-y))$
where both $(\alpha, \beta)$ and $(\varphi, \psi)$ are in $H^+$.
Then, we have $G \circ F(x,y) = \frac{1}{2}(\alpha \circ
\varphi(x+y)+\beta \circ \psi(x-y), \alpha \circ
\varphi(x+y)-\beta \circ \psi(x-y))$ and thus $\Pi(G \circ F) =
(\alpha \circ \varphi, \beta \circ \psi)$. Since $\Pi(G)=(\alpha,
\beta)$ and $\Pi(F) = (\varphi, \psi)$, we have $\Pi(G \circ
F)=\Pi(G) * \Pi(F)$.\\

ii) Let
$G=\frac{1}{2}(\alpha(x+y)+\beta(x-y),\alpha(x+y)-\beta(x-y))$ and
$F=\frac{1}{2}(\varphi(x-y)+\psi(x+y), \varphi(x-y)-\psi(x+y))$
where $(\alpha, \beta)$ is in $H^+$ and $(\varphi, \psi)$ is in
$H^-$. Then, we have $G \circ F(x,y) = \frac{1}{2}(\alpha \circ
\varphi (x-y)+\beta \circ \psi(x+y), \alpha \circ
\varphi(x-y)-\beta \circ \psi(x+y))$ and thus $\Pi(G \circ
F)=(\alpha \circ \varphi, \beta \circ \psi$). Since
$\Pi(G)=(\alpha, \beta)$ and $\Pi(F)=(\varphi,\psi)$, we have
$\Pi(G \circ F) = \Pi(G) * \Pi(F)$.\\

iii) Let $G=\frac{1}{2}(\alpha(x-y)+\beta(x+y),
\alpha(x-y)-\beta(x+y))$ and
$F=\frac{1}{2}(\varphi(x+y)+\psi(x-y), \varphi(x+y)-\psi(x-y))$
where $(\alpha, \beta)$ is in $H^-$ and $(\varphi, \psi)$ is in
$H^+$. Then, we have $G \circ F(x,y) = \frac{1}{2}(\alpha \circ
\psi(x-y)+\beta \circ \varphi(x+y), \alpha \circ
\psi(x-y)-\beta\circ\varphi(x+y))$ and thus $\Pi(G \circ
F)=(\alpha \circ \psi, \beta \circ \varphi)$. Since $\Pi(G) =
(\alpha, \beta)$ and $\Pi(F)=(\varphi, \psi)$, we have $\Pi(G
\circ F)= (\alpha\circ\psi,\beta\circ\varphi) = (\alpha,\beta)
\circ \pi(\alpha,\beta)\cdot(\varphi,\psi)=\Pi(G)*\Pi(F)$.\\

iv) Let $G=\frac{1}{2}(\alpha(x-y)+\beta(x+y),
\alpha(x-y)-\beta(x+y))$ and
$F=\frac{1}{2}(\varphi(x-y)+\psi(x+y), \varphi(x-y)-\psi(x+y))$
where both $(\alpha, \beta)$ and $(\varphi,\psi)$ are in $H^-$.
Then, we have $G \circ F(x,y) = \frac{1}{2}(\alpha \circ
\psi(x+y)+\beta\circ\varphi(x-y), \alpha \circ \psi(x+y)-\beta
\circ \varphi(x-y))$ and thus $\Pi(G \circ F)=(\alpha \circ \psi,
\beta \circ \varphi)$. Since $\Pi(G)=(\alpha, \beta)$ and
$\Pi(F)=(\varphi,\psi)$, we have $\Pi(G \circ F) =
(\alpha\circ\psi, \beta\circ\varphi)=(\alpha,\beta) \circ
\pi(\alpha,\beta) \cdot (\varphi, \psi) = \Pi(G) * \Pi(F)$.

This show that $\Pi$ is an isomorphism and the proof is completed.

\end{proof}

In Ref. \cite{CQG2}, it is shown that $H(\mathbb{R})$ is a
subgroup of $G$ and this can also be seen in the above theorem as
follows. If we define a map $\Omega : H(\mathbb{R}) \rightarrow H$
by $\Omega(f)=(f,f)$, then it is easy to see that $\Omega$ is an
injective homomorphism and thus, $H(\mathbb{R})$ is a subgroup of
$G$ through an injective homomorphism $\Pi^{-1} \circ \Omega$.
Zeeman's result tells us that the group of causal automorphisms on
$\mathbb{R}^n_1$ is finite dimensional when $n \geq 3$ and our
result tells us that the group is infinite dimensional when $n=2$.

\section{acknowledgement}

This work was supported by a grant from the College of Applied
Science, Kyung Hee University research professor fellowship.

\end{document}